\documentclass{article}
\usepackage[draft]{hyperref} 
\usepackage{setspace}
\usepackage{amsmath}
\usepackage{amssymb}
\usepackage{amsthm}
\usepackage{multirow}
\usepackage{amsfonts}
\usepackage{color}
\usepackage{graphicx}
\usepackage{subfigure}  
\usepackage[]{algorithmicx}
\usepackage{algpseudocode,algorithm}
\usepackage{cite}
\usepackage{mathtools}
\usepackage{stmaryrd}
\usepackage[mathscr]{euscript}
\usepackage{mathrsfs}

\usepackage{soul}

\newtheorem{theorem}{Theorem}
\newtheorem{lemma}{Lemma}

\newtheorem{corollary}{Corollary}
\usepackage[normalem]{ulem}
\usepackage{spconf,amsmath,graphicx}
\usepackage{bbm}
\usepackage{dsfont}

\definecolor{DarkGreen}{rgb}{0.1,0.5,0.1}
\definecolor{DarkRed}{rgb}{0.5,0.1,0.1}
\definecolor{DarkBlue}{rgb}{0.1,0.1,0.5}
\definecolor{DarkPurple}{rgb}{0.5,0.2,0.5}
\definecolor{DarkTurquoise}{rgb}{0.1,0.5,0.5}

\definecolor{beaublue}{rgb}{0.74, 0.83, 0.9}
\definecolor{coolblack}{rgb}{0.0, 0.18, 0.39}
\definecolor{apricot}{rgb}{0.98, 0.81, 0.69}
\definecolor{burntorange}{rgb}{0.8, 0.33, 0.0}
\definecolor{blue-violet}{rgb}{0.54, 0.17, 0.89}
\definecolor{byzantium}{rgb}{0.44, 0.16, 0.39}
\definecolor{brilliantrose}{rgb}{1.0, 0.33, 0.64}
\definecolor{cerisepink}{rgb}{0.93, 0.23, 0.51}
\definecolor{cobalt}{rgb}{0.0, 0.28, 0.67}
\definecolor{bostonuniversityred}{rgb}{0.8, 0.0, 0.0}

\newcommand{\off}[1]{}

\title{Compressed IF-TEM: Time Encoding Analog-to-Digital Compression}
\name{Saar Tarnopolsky$^{*\ddag}$, Hila~Naaman$^{\dag\ddag}$, Yonina C. Eldar$^{\dag}$, and Alejandro~Cohen$^{*}$
\thanks{\scriptsize $^{\ddag}$Equal contribution.}
\thanks{\scriptsize{This research was partially supported by a research grant from the Estate of Tully and Michele Plesser, by the European Union’s Horizon 2020 research and innovation program EU-H2020-ERC-CoG under grant No. 101000967, by the Israel Science Foundation under grant No. 536/22.}}}
\address{$^{*}$Electrical and Computer Engineering, Technion, Israel, \{saar.tar, alecohen\}@technion.ac.il\\
$^{\dag}$Math and CS, Weizmann Institute of Science, Israel, \{hila.naaman, yonina.eldar\}@weizmann.ac.il}
\begin{document}
\ninept
\maketitle
\begin{abstract}
An integrate-and-fire time-encoding-machine (IF-TEM) is an energy-efficient asynchronous sampler.
Utilizing the IF-TEM sampler for bandlimited signals, we introduce designs for time encoding and decoding with analog compression prior to the quantization phase.
Before the quantizer, efficient analog compression is conducted based on the stationarity of the encoded signal, which is a fundamental characteristic of IF-TEM processing.
Low-bit-rate reconstruction is achieved by subdividing the known IF-TEM dynamic range into tighter windows, which can be either fixed size or dynamically changed, and detecting in which window the sample resides. We demonstrate empirically that employing the same number of samples and up to 7$\%$ additional bits than the conventional IF-TEM results in a 5-20dB improvement in MSE.
Fixing the reconstruction MSE target and the number of samples, using the compressed IF-TEM enables the use of 1-2 fewer bits compared to the classical IF-TEM.
\end{abstract}
\section{Introduction}
Analog to Digital Converters (ADCs) are among the most utilized components in digital systems today. The majority of today's samplers require highly accurate clock support \cite{miskowicz2006dynamic}. It is considered a difficult task to design highly precise clock samplers \cite{koscielnik2007designing}. As a result, asynchronous sampling systems began to gain favor as a suitable choice for successfully replacing synchronous clock samplers \cite{miskowicz2018event,1199179}.
Asynchronous samplers pave the way toward robust and energy-efficient analog sampling circuits. These systems are typically based on changes in the input signal rather than sampling at regular intervals \cite{miskowicz2018event}. In contrast to synchronous sampling circuits, asynchronous systems can encode time intervals only rather than time-amplitude pairs \cite{feichtinger2012approximate,adam2020sampling}.

Integrate-and-fire time-encoding-machine (IF-TEM) is one of the most prominent time-based asynchronous energy-efficient samplers \cite{lazar2004perfect,lazar2004time,adam2020sampling,alexandru2019reconstructing,naaman2022fri}. The signal is integrated, and when it crosses a given threshold, the time is encoded. IF-TEM sampler is known to be low-power and operates very similarly to neurons in the human body \cite{lazar2004perfect,lazar2004time}.
In recent years, the IF-TEM sampler has been used for sampling and reconstruction of a variety of signal classes, including band-limited (BL) \cite{lazar2004time,lazar2011video,adam2020sampling,thao2020time,adam2021asynchrony,naaman2021time}, finite-rate-of-innovation (FRI) \cite{naaman2022fri,naaman2022uniqueness,rudresh2020time}, and shift-invariant signals \cite{gontier2014sampling}. Quantification's effect on the quality of reconstruction utilizing an IF-TEM is also addressed \cite{lazar2004perfect,naaman2021time}.
However, prior studies did not leverage the IF-TEM's natural gain as a low-pass filter (LPF). In particular, the IF-TEM outputs are stationary, which permits efficient analog compression throughout the sampling period before quantization.

In this paper, we present a compressed IF-TEM (CIF-TEM) sampler, a design methodology for analog compression, in which analog bandlimited (BL) signals are sampled and quantized at the Nyquist rate using the IF-TEM ADC.
Based on our prior work \cite{naaman2021time}, we analyze BL signals that have finite energy, and their amplitudes that are constrained by the frequency and energy \cite{papoulis1967limits}.
Our method relies on subdividing the known IF-TEM dynamic range into tighter sub-ranges with are then quantized. Leveraging the stationarity of the IF-TEM output, the samples would be analogy compressed prior to quantization. Compared to the IF-TEM without compression using the same number of samples, the reconstruction MSE is improved while the total number of bits is reduced.
To the best of our knowledge, this is the first work to establish analog compression in IF-TEM systems.

Our contribution is threefold: first, we provide an analog compression technique based on the stationarity of the samples. By subdividing the known IF-TEM dynamic range into fixed-size narrower windows, this method enables efficient analog and digital implementation as well as low bit-rate reconstruction.
Second, we incorporate an adaptive component in the compression process so that it can better adapt to the sampler system, resulting in compression savings of at least 1-2 bits compared to sampling and recovering BL signals using IF-TEM quantization without an adaptive component.
Thirdly, we demonstrate numerically that by
utilizing our proposed time-based analog compression technique, the MSE of the reconstructed BL signal is lower than that proposed by \cite{naaman2021time} while using the same amount of samples and fewer bits.

The rest of the paper is organized as follows. In Section \ref{sec:prblem_for}, we discuss some background results and formulate the problem. In Section \ref{sec:CIF-TEM}, we present the compressed IF-TEM (CIF-TEM) scheme. Simulation results are provided in Section \ref{sec:simulation}, followed by conclusions.
\section{PROBLEM FORMULATION AND PRELIMINARIES}
\label{sec:prblem_for}
We begin by presenting relevant known results on IF-TEM sampling, perfect reconstruction, and quantization, followed by the problem formulation.
\subsection{IF-TEM for BL signals}
\label{section:SIRec}
An IF-TEM sampler is characterized by a bias $b$, scaling $\kappa$, and threshold $\delta$, as depicted in Fig. \ref{fig:TEM_quant}.
The input $x(t)\in L^2(\mathbb{R})$ is assumed to fulfill the condition $|x(t)|\leq c<b<\infty$.  As shown in Fig. \ref{fig:IF-TEM_scheme}, before time-encoding the signal $x(t)$, a bias $b$ is added. The resultant signal $b+x(t)$ is integrated after being scaled by $1/\kappa$. Either the time instances $\{t_n\}_{n\in\mathbb{Z}}$ or their differences $T_n = t_{n+1} - t_{n}$, at which the integral crosses the threshold $\delta$ are recorded, and the integrator is reset to zero.\\
The following relationship between system parameters:
\begin{align}
   \frac{1}{\kappa}\int_{t_n}^{t_{n+1}} (b+x(\tau))\, d\tau = \delta.
\end{align}
Alternatively, the following amplitude measurements can be computed:
\begin{equation}
x_n \triangleq\int_{t_n}^{t_{n+1}}x(\tau)\, d\tau =
  -b(t_{n+1}-t_n)+\kappa\delta.
\label{eq:trigger0}
\end{equation}
Using \eqref{eq:trigger0} and \cite{lazar2004perfect,lazar2004time} it can be shown that $T_n$ is bounded by:
\begin{equation}
\Delta t_{min}\triangleq\frac{\kappa\delta}{b+c}\leq T_n\leq\frac{\kappa\delta}{b-c}\triangleq\Delta t_{max}.
\label{eq:bounds}
\end{equation}
The reconstruction of BL signals from IF-TEM outputs has been studied for scenarios when the input signal is $2\Omega$ BL and $c$-bounded with finite energy $E$, where $\Omega$ is the frequency upper bound \cite{naaman2021time,adam2020sampling,lazar2004perfect,lazar2004time}. As in \cite{naaman2021time}, we assume BL signals with maximal energy $E$; in this case, the relation between $\Omega$ and the maximal signal's amplitude $c$ is $ c = \sqrt{{E\Omega}/{\pi}}$, as shown in \cite{papoulis1967limits}.

\begin{figure}
	\centering
	\includegraphics[trim=1cm 0cm 0cm 0cm, clip, width = 1.05\columnwidth]{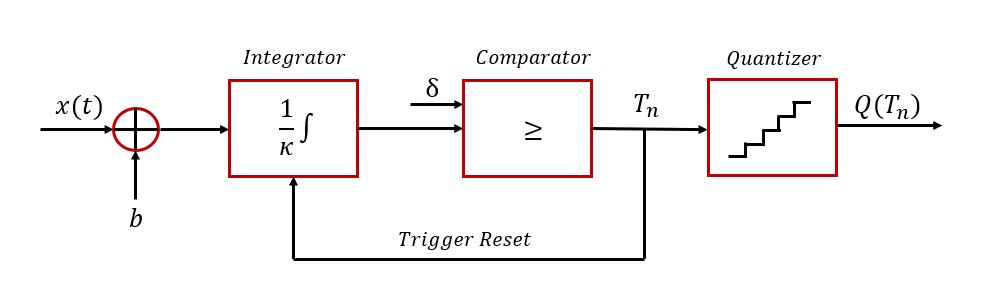}
	\vspace{-0.5cm}
	\caption{An IF-TEM time instances differences $T_n$ sampling and quantization scheme.}
		\label{fig:TEM_quant}
\end{figure}

We consider the IF-TEM sampling and reconstruction technique as in \cite{lazar2004time}, which showed that BL can be perfectly recovered using an IF-TEM with parameters $\{b,\kappa,\delta\}$ if $b>c$ and \cite{lazar2004time}
\begin{equation}
    \Delta t_{\max} < \frac{\pi}{\Omega}.
    \label{eq:density}
\end{equation}
The bound in \eqref{eq:density} necessitates a bandwidth that is inversely proportional to the time difference between the time instances. Therefore, the BL input signal can be perfectly recovered if the overall firing rate of the IF-TEM is greater than the Nyquist rate.
In this case, the signal is reconstructed similarly to a BL signal recorded with irregularly spaced amplitude samples (see \cite{lazar2004time,adam2020sampling} for information on the IF-TEM recovery techniques).

Using an IF-TEM as depicted in Fig. \ref{fig:TEM_quant}, the authors in \cite{lazar2004perfect} studied the MSE theoretical bounds of the reconstruction of signals sampled using IF-TEM and uniform quantizer. The authors in \cite{naaman2021time} studied the problem of quantizing the time instances differences samples $T_n$ of BL signals. They considered a $K$-level uniform design for the quantization levels and demonstrated that the quantization using IF-TEM results in better reconstruction for BL signals compared to conventional uniform sampling.
They proved that if the bias $b$ is chosen such that $b = \alpha c$, where $\alpha>1$, the step-size is given by
\begin{equation}
    \Delta_{\text{IF-TEM}} = \frac{\kappa\delta}{(\alpha+1)(\alpha-1)} \frac{2}{cK}.
    \label{eq:iftem_stepsize2}
\end{equation}

In particular, by experimental study, they showed that the quantizer step size as well as dynamic range  $\left[\frac{\kappa \delta}{b + c},\frac{\kappa \delta}{b-c}\right]$ of the samples $T_n$ decrease, resulting in a lower MSE compared with the amplitude Nyquist sampling and reconstruction method. In addition, they demonstrated that step size decreases as signals bandwidth increases.

\subsection{Problem Statement}
\label{sec:problem_formulation}
We consider the problem of recovering a $2\Omega$ BL and $c$-bounded signal  $x(t)\in L^2(\mathbb{R})$ with a fixed maximal energy $E<\infty$, from its quantized samples $T_n = t_{n+1}-t_n$ for the IF-TEM setup.
According to \cite{papoulis1967limits}, such signals are amplitude bounded and meet the following condition: \begin{equation}
    |x(t)| \leq c = \sqrt{E\Omega/\pi}.
\end{equation}

A generalized CIF-TEM sampling with quantization scheme is shown in Figures \ref{fig:CIF-TEM} and \ref{fig:CIF-TEM2}. The IF-TEM sampler with parameters $\{b,\kappa,\delta\}$ computes the stationary time-based measurements $T_n = t_{n+1}-t_n$ of the signal $x(t)$, which are then quantized using the compression and estimator blocks, resulting in $Q(r_n) = \hat{r}_n$, where $Q$ is the quantizer, $r_n$ and $\hat{r}_n$ denotes $T_n$ compressed and then quantized, respectively.
Due to quantization, it is impossible to fully recover the signal $x(t)$. The authors of \cite{naaman2021time} demonstrated that the IF-TEM sampler can result in a superior MSE of the recovered BL signal from quantized measurements as compared to the standard Nyquist reconstruction.
\begin{figure}
	\centering
	\includegraphics[width=0.5\textwidth]{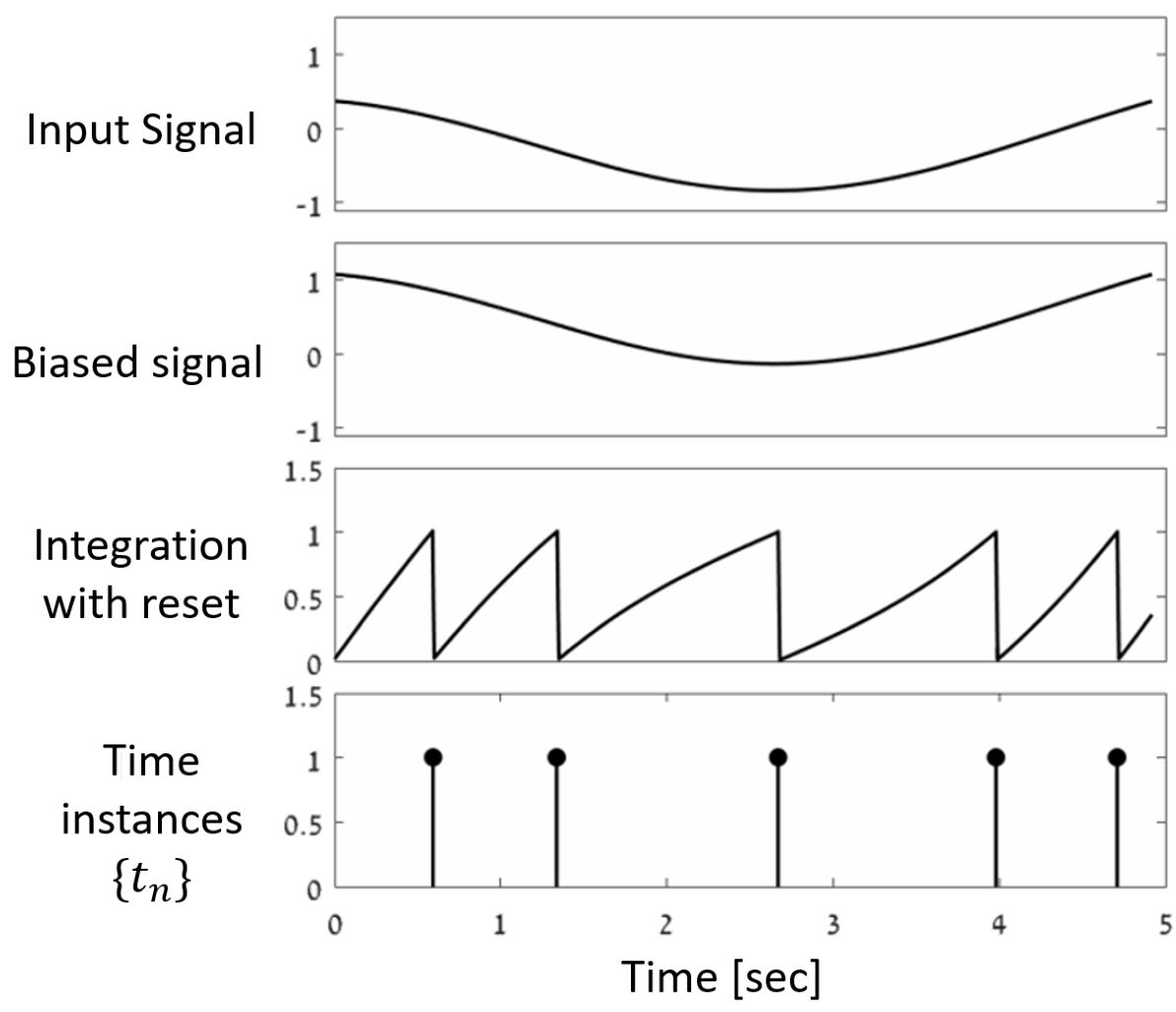}
	\vspace{-0.5cm}
	\caption{The IF-TEM system. The signal is biased and then integrated, each time the threshold is reached the time instances or their differences are recorded. }
	\label{fig:IF-TEM_scheme}
\end{figure}

We study the problem of recovering a signal from its quantized samples for both the conventional IF-TEM with quantization system proposed by \cite{naaman2021time} and our proposed CIF-TEM method, which employs compression at the quantization phase (see Fig. \ref{fig:CIF-TEM}).
We would like to establish the advantages of CIF-TEM over IF-TEM with quantization proposed by \cite{naaman2021time} in terms of MSE and bits.
\section{CIF-TEM: compressed IF-TEM}
\label{sec:CIF-TEM}
In this section, we analyze the problem of recovering a signal from its quantized samples for both the conventional IF-TEM with quantization system proposed by \cite{naaman2021time} and our proposed CIF-TEM method, which uses analog compression before the quantization phase.
Prior to the quantization operation for the IF-TEM, an estimator and compression blocks are added.  The advantages of CIF-TEM stem from the fact that the IF-TEM sampler acts as a LPF and, as such, its output is stationary in time, the measured sample values are close to one another, and analog compression can be efficiently applied. We propose two methods for dividing the dynamic range of quantized samples $\hat{T}_n$ into time windows. We determine the window number, which may be constant or fluctuate dynamically. We refer to these approaches, as constant CIF-TEM (CCIF-TEM) and dynamic CIF-TEM (DCIF-TEM), respectively.

As depicted in Fig. \ref{fig:CIF-TEM}, the IF-TEM output $T_n$ is the estimator block input. The estimator block calculates $L$, the number of windows $\{w_i\}_{i=0}^{L-1}$ which divides the dynamic range of the encoded signal
\[
\Delta_t \triangleq \Delta t_{max}-\Delta t_{min},
\]
for higher quantization resolution.
The compression block is then divides the dynamic range into $L$ windows, and determines to which of the windows $\{w_i\}_{i=1}^L$ the sample $T_n$ belongs. Then $K$ bits are assigned to quantize the sample $T_n$ in that specific sub-dynamic range regime $r_n$. The output of the compression block $r_n$, represents the sub dynamic range regime of the window $\{w_i\}_{i=0}^{L-1}$. Finally, the quantizer quantize  $r_n$. By decreasing the quantization step size and employing sub-dynamic range regimes, the reconstruction is improved in terms of MSE utilizing the CIF-TEM compared to \cite{naaman2021time,lazar2004perfect}.
Note that the window number $w_i$ is encoded separately from the samples. Since it is assumed that the encoded samples are stationary in time, the bulk of consecutive samples will fall within the same time windows. Consequently, the value of $w_i$ will fluctuate every few samples. Thus, $w_i$ is only encoded if its value has changed since the last sample.
Despite the fact that CIF-TEM techniques require more bits for window encoding, as demonstrated numerically in the next section, they can achieve the same MSE as IF-TEM with quantization while using less bits overall.  An example of a practical implementation for the analog compression scheme in CIF-TEM is illustrated in Fig. \ref{fig:CIF-TEM2}.

\subsection{CCIF-TEM}
\label{ssec:IF-TEMW with constant windows}
In the CCIF-TEM scheme, the number of windows $L\in \mathbb{N}$ is a positive constant number.
The sample's $T_n$ statistics are assumed to be known constants. Namely, the sample's variance $\sigma$ and mean $\mu$.
Given the sample's variance, the constant number of windows is determined as
\begin{equation}
    L = \left\lceil\frac{\Delta t_{max}-\Delta t_{min}}{2\sqrt{\sigma}}\right\rceil.
    \label{eq:L}
\end{equation}

\begin{lemma}
Consider a 2$\Omega$ $c$-bounded BL signal and an CCIF-TEM sampler with parameters $\{\kappa,\delta,b\}$, such as $b=\alpha c$, where $\alpha>1$. Let $\{\sigma,\mu\}$ be the known variance and mean samples statistics. Using a $K$-level uniform quantizer, the CCIF-TEM quantization step size is
\begin{equation}
     \Delta_{CCIF-TEM}=\frac{\kappa \delta}{(\alpha-1)(\alpha+1)}\frac{2}{cLK}.
    \label{eq:D_CCIF}
\end{equation}
\end{lemma}
\begin{proof}
Using \eqref{eq:bounds}, the size of the samples $T_n$ dynamic range is:
\begin{equation}
    \Delta t_{max}-\Delta t_{min} = \frac{\kappa\delta}{(b-c)(b+c)}2c.
\end{equation}
Since the dynamic range size is divided into $L$ uniform sub-ranges and then quantized using $K$-level quantizer, the CCIF-TEM quantization step size is
\begin{equation}
    \Delta_{CCIF-TEM}=\frac{\kappa\delta}{(b-c)(b+c)}\frac{2c}{LK}.
\end{equation}
Using the relation $b=\alpha c$, where $\alpha>1$ results
\begin{equation*}
    \Delta_{CCIF-TEM}=\frac{\kappa \delta}{(\alpha-1)(\alpha+1)}\frac{2}{cLK}.
\end{equation*}
\end{proof}
\begin{figure}
	\centering
	\includegraphics[trim=1.1cm 0cm 0cm 0cm, clip, width = 1.11\columnwidth]{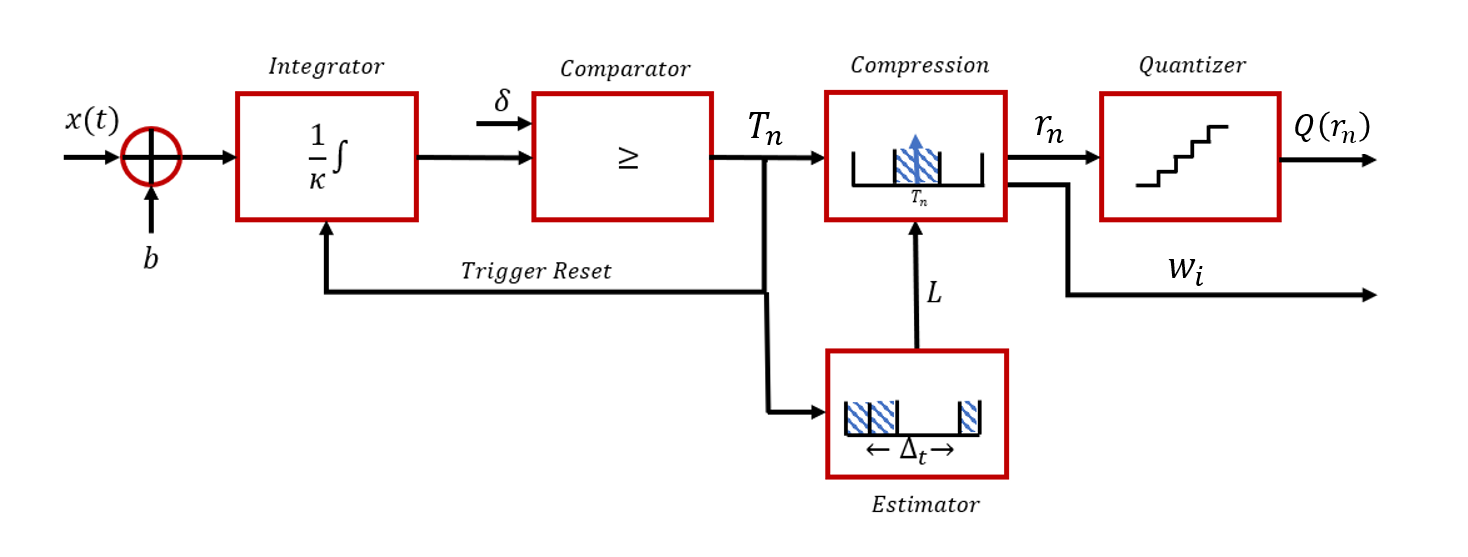}
	\vspace{-0.5cm}
	\caption{The CIF-TEM system. In the quantization phase, it consists of a compression block and an estimator block for determining the number of partitions to the dynamic range of the samples.}
		\label{fig:CIF-TEM}
\end{figure}
\begin{figure}
	\centering
	\includegraphics[width=0.5\textwidth]{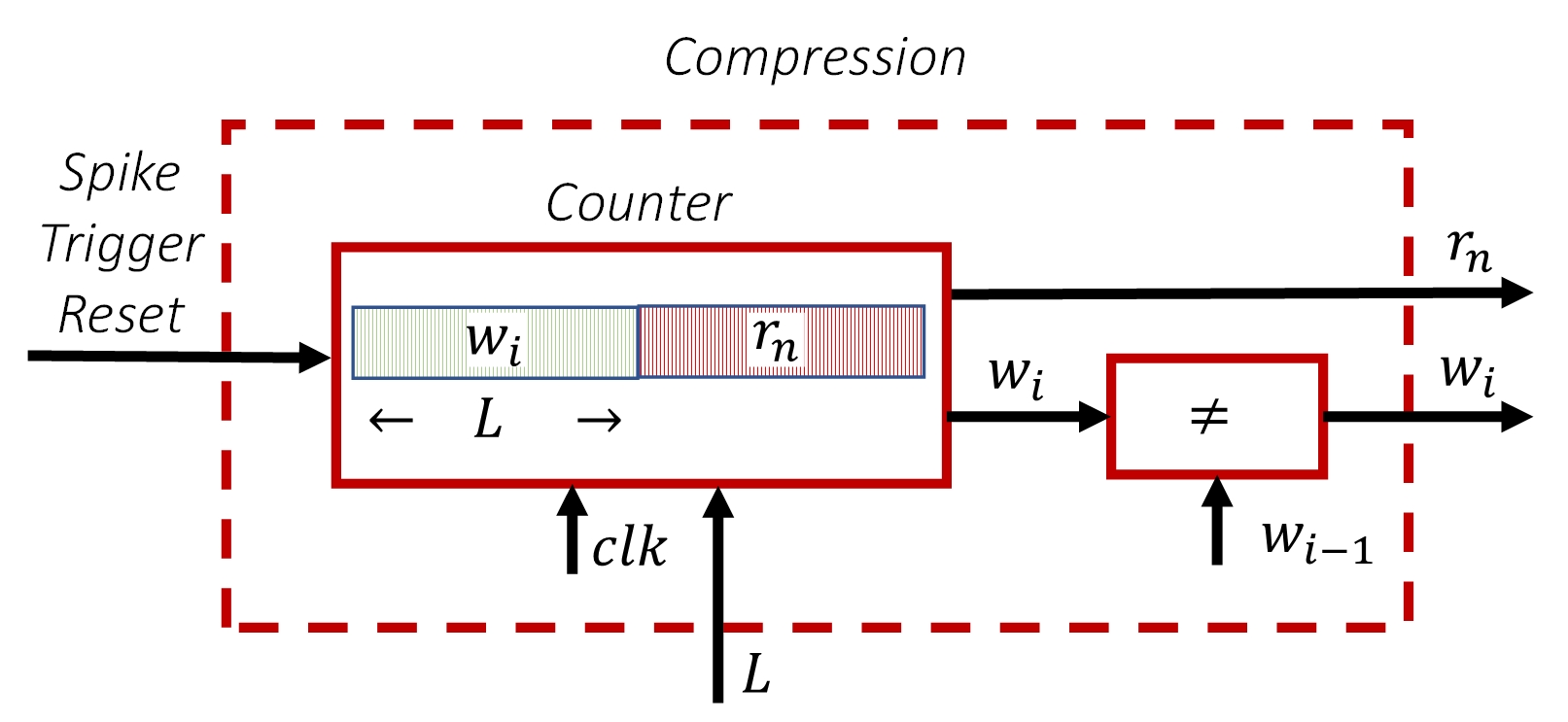}
	\vspace{-0.5cm}
	\caption{Analog compression scheme in CIF-TEM. The counter is counting the time instances differences $T_n$. The higher $L$ bits of the counter are assigned to represent window number $w_i$, and the entire lower bits of the counter are assigned to represent sub-dynamic range regime $r_n$. When the integral of the IF-TEM scheme crosses the threshold $\delta$ (see Fig.~\ref{fig:CIF-TEM}), the value of the counter is recorded and then reset. That is, $r_n$ is quantized by the quantizer in the CIF-TEM system, and $w_i$ is compared to the previously recorded value $w_{i-1}$. If the $w_i \neq w_{i-1}$, the value of the new window number is recorded.}
		\label{fig:CIF-TEM2}
\end{figure}


Note that $L$ and $\sigma$ are inversely proportional to each other. Using \eqref{eq:L} and \eqref{eq:D_CCIF}, it follows that decreasing $\sigma$ will increase $L$, which will decrease the quantization step size $\Delta_{CCIF-TEM}$.
Next, we demonstrate that when $L>1$, the quantization step size of the CCIF-TEM is smaller than that of the IF-TEM, hence potentially decreasing the MSE. This outcome is summarized in the following theorem.

\begin{theorem}
Let $\Delta_{IF-TEM}$ be the quantization step size of an IF-TEM sampler with a $K$-level uniform quantizer proposed by \cite{naaman2021time}. Let $\Delta_{CCIF-TEM}$ be the quantization step size of a CCIF-TEM sampler with a $K$-level uniform quantizer. Both samplers have the same parameters $\{\kappa,\delta,b\}$ and the same uniform quantizer with $K$ quantization levels.
Let the sample's variance $\sigma$ be a constant positive number.
In this case, $\Delta_{CCIF-TEM} < \Delta_{IF-TEM}$.
\end{theorem}
\begin{proof}
Using Popoviciu inequality \cite{popoviciu1935equations}, it follows that \begin{equation}
    {\sigma} < \frac{\left(\Delta t_{max}- \Delta t_{min}\right)^2}{4}.
\end{equation}
Thus,
$2\sqrt{\sigma} < \Delta t_{max}- \Delta t_{min}$. In this case, using \eqref{eq:L}, it implies that $L>1$.
Using Lemma 1 with \eqref{eq:iftem_stepsize2} and \eqref{eq:D_CCIF}, t can be shown that $L \times \Delta_{CCIF-TEM} = \Delta_{IF-TEM}$.
Since $L>1$ it follows that $\Delta_{CCIF-TEM} < \Delta_{IF-TEM}$.
\end{proof}
Note that quantization error can be reduced based on dense quantization, and the CCIF-TEM framework results in a lower quantization error than the IF-TEM scheme.
Since the CCIF-TEM compression utilizes a constant $L$ using the signal statistics, we want to be able to estimate the statistics in real time and adjust $L$ accordingly, as we suggest in the following sub-section.

\subsection{DCIF-TEM}
\label{ssec:IF-TEMW with dynamic windows}
In this section, we present the DCIF-TEM approach, which permits the window size $L$ to alter dynamically during sample acquisition by estimating the signal statistics in real-time.
In contrast to the CCIF-TEM, where the constant number of windows is determined using previous knowledge about the input signal, we do not have access to this information in the DCIF-TEM.
The system is initiated with a default value for $L$, and is converging time.

In particular, the estimator block computes $L$ for each fixed $m>1$ IF-TEM samples based on the variance of the previous $m$ samples. The variance is denoted by $\hat{\sigma}$ and is given as:
\begin{equation}
    \hat{\sigma}=\frac{1}{m}\sum_{i=n-m+1}^{n} (T_{i}-\hat{\mu_m})^{2},
\end{equation}
where $\hat{\mu_m}$ is the the average of the last $m$ samples and is calculated as:
\begin{equation}
    \hat{\mu}=\frac{1}{m}\sum_{i=n-m+1}^{n} T_{i}.
\end{equation}
In the DCIF-TEM scheme, the number of windows $L$ is given as
\begin{equation}
    L = \left\lceil\frac{\Delta t_{max}-\Delta t_{min}}{2\sqrt{\hat{\sigma}}}\right\rceil.
\end{equation}
Note that
increasing $\hat{\sigma}$ decreases $L$, which increases the quantization step size $\Delta_{DCIF-TEM}$.
Similar to the CCIF-TEM sampler, we show that when $L>1$, the quantization step size of the DCIF-TEM is smaller than that of the IF-TEM, potentially reducing the MSE.
\begin{theorem}
Let $\Delta_{IF-TEM}$ be the quantization step size of an IF-TEM sampler and $\Delta_{DCIF-TEM}$ be the quantization step size of a DCIF-TEM sampler. Let also both samplers have the parameters $\{\kappa,\delta,b\}$ and the same uniform quantizer $Q$ with $K$  quantization levels.
Then $\Delta_{DCIF-TEM} < \Delta_{IF-TEM}$.
\end{theorem}
\begin{proof}
Using \eqref{eq:bounds} and Popoviciu inequality in \cite{popoviciu1935equations}, implies that the variance of any set of consecutive samples $\{T_i\}_{i=n-m+1}^{n}$ such that $n>m$ upholds
\[
2\sqrt{\hat{\sigma}} < \Delta t_{max} - \Delta t_{min}.
\]
Since $L = \left\lceil\frac{\Delta t_{max}-\Delta t_{min}}{2\sqrt{\hat{\sigma}}}\right\rceil$ we conclude that $L>1$. The dynamic range of the IF-TEM sampler is defined as $\frac{2c\kappa\delta}{(b-c)(b+c)}$ and the quantization step size is $\frac{\kappa\delta}{(b-c)(b+c)}\frac{2c}{K}$. DCIF-TEM works with sub ranges each of size $\frac{\kappa\delta}{(b-c)(b+c)}\frac{2c}{L}$ such that the quantization step size is
\[
\frac{\kappa\delta}{(b-c)(b+c)}\frac{2c}{LK}.
\]
Following that we conclude that
\[
L \times \Delta_{CCIF-TEM} = \Delta_{IF-TEM}.
\]
Since $L>1$ the condition $\Delta_{DCIF-TEM} < \Delta_{IF-TEM}$ upholds.
\end{proof}
Note that since the samples are stationary, we expect that the size of the window will not vary much between samples. Thus, $\hat{\sigma}$ is computed and updated every $l$ samples, where $l>1$.
Thus, the number of windows $L$ only changes when the variance, $\sigma$, varies significantly.
Signal statistics are computed dynamically in DCIF-TEM, whereas they are assumed in CCIF-TEM. Knowing rather than estimating the exact sample statistics results in a reduced MSE. Consequently, the CCIF-TEM has a lower MSE than the DCIF-TEM, which contains less prior information.

\begin{corollary}
Based on Lemma 1 and Theorems 1 and 2, in order to reach $\Delta_{CCIF-TEM} = \Delta_{IF-TEM}$ or $\Delta_{DCIF-TEM} = \Delta_{IF-TEM}$, one can use uniform quantizer $Q'$ with $K'$ quantization levels in the CCIF-TEM and DCIF-TEM samplers such that $K'L = K$.
\end{corollary}

Since the number of bits utilized to store the signal is directly proportional to $\log_2K$, the CIF-TEM systems compress signal bits more than IF-TEM systems. Note that CIF-TEM systems also encode window numbers, which introduces an additional piece of cost. It is assumed that we know the sample's statistics or can dynamically estimate them. Therefore we expect that consecutive samples will belong to the same time windows, resulting in a low number of window number changes. Our simulations demonstrate that using fewer quantization levels in the quantizers of CIF-TEM systems maintains the same MSE as an IF-TEM sampler while performing bit compression.
In future work, we will drive analytical results for this overhead.

\section{Evaluation results}
\label{sec:simulation}
In this section, we exemplify our main result in an experimental study using simulations. First, we demonstrate that the CIF-TEM schemes outperform the IF-TEM quantization approach proposed in \cite{naaman2021time} in terms of MSE, by employing the same amount of IF-TEM samples and overall bits.
Secondly, we show that CCIF-TEM produces a lower MSE than DCIF-TEM. It is shown, however, that by using the DCIF-TEM rather than the CCIF-TEM to encode the same amount of samples, less bits may be employed, resulting in the same MSE.

We verify our main result using a BL signal $x(t)$ as input. We consider 100 randomly selected $2\Omega$ bandlimited signals $x(t)$, which is bounded in time, i.e., $|x(t)|\leq c$, for $c=\sqrt{(E\Omega)/\pi}$ with maximal $E=0.8$ and $\Omega$ varying from $5-80$ Hz. The quantization bit range $\log_2K \in \{6,7,...,15\}$.
We investigate the recovery after quantization for the IF-TEM and the CIF-TEM methods. The IF-TEM parameters are selected as follows; we use fixed values of $\delta=0.075$ and $\kappa=0.5$. To have a sufficient number of samples needed for recovery, the bias is selected in two ways resulting in a maximal oversampling factor of 3.5: first, a fixed $b=35$ for all signals according to \cite[Section 2]{lazar2004perfect}, and second, choosing $b=\alpha c$ with $\alpha=6$ according to \cite[Section 3]{naaman2021time}.

In Fig. \ref{fig:res1}, the suggested CIF-TEM sampling frameworks with quantization (CCIF-TEM and DCIF-TEM) are evaluated in terms of MSE and compared to the IF-TEM approach suggested by \cite{naaman2021time}.
The MSE is calculated as: \begin{equation}
MSE=20\log_{10}||x(t)-\hat{x}(t)||_{L_2[0,T]}.
\end{equation}
As can be shown, better MSE is reached using less bits in CIF-TEM systems compared to the IF-TEM systems.
The CCIF-TEM surpasses the DCIF-TEM by employing fewer quantization bits for better reconstruction for the same number of samples.
This result may be related to the fact that in the CCIF-TEM, the constant number of windows is determined using prior information about the input signal, whereas in the DCIF-TEM, this information is unavailable.
The DCIF-TEM parameters are $l=5$ and $m=40$. Using the DCIF-TEM sampler, the number of windows $L$ is calculated as the average number of windows used to encode the signal, $2 \leq L \leq 8$, i.e., 2-3 bits are needed to encode the window.

\begin{figure}[t]
\begin{minipage}[b]{1.0\linewidth}
  \centering
  \centerline{\includegraphics[trim=0cm 0cm 0cm 0cm, clip, width = 1.05\columnwidth]{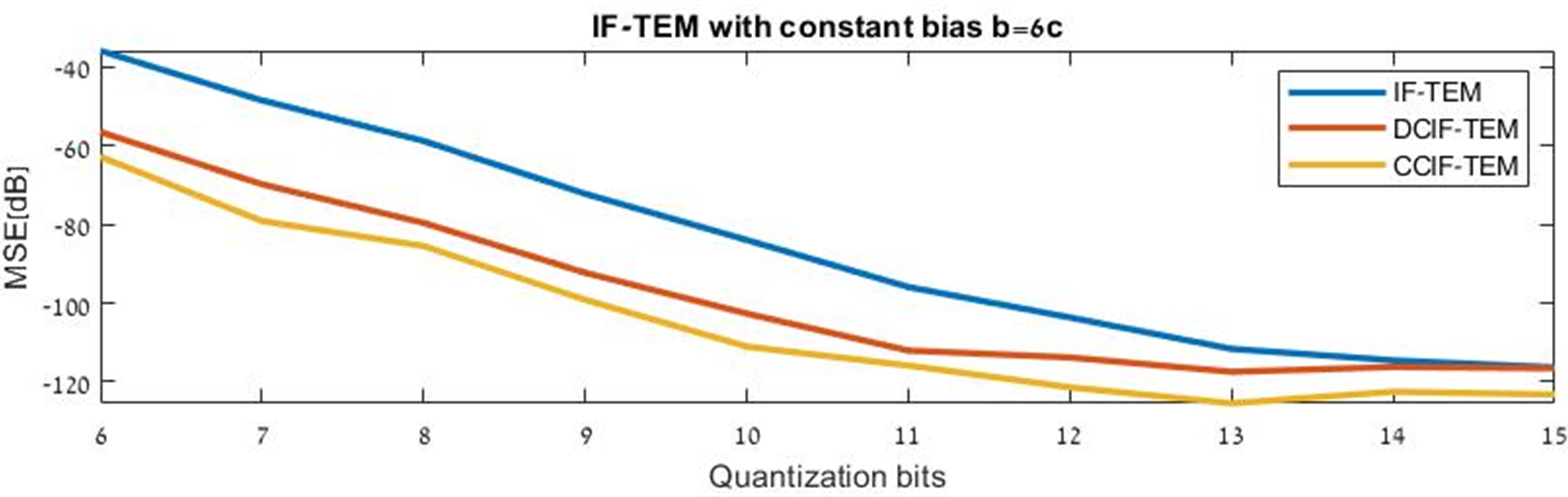}}
  \centerline{(a)}\medskip
\end{minipage}
\begin{minipage}[b]{1.0\linewidth}
  \centering
  \centerline{\includegraphics[trim=0cm 0cm 0cm 0cm, clip, width = 1.05\columnwidth]{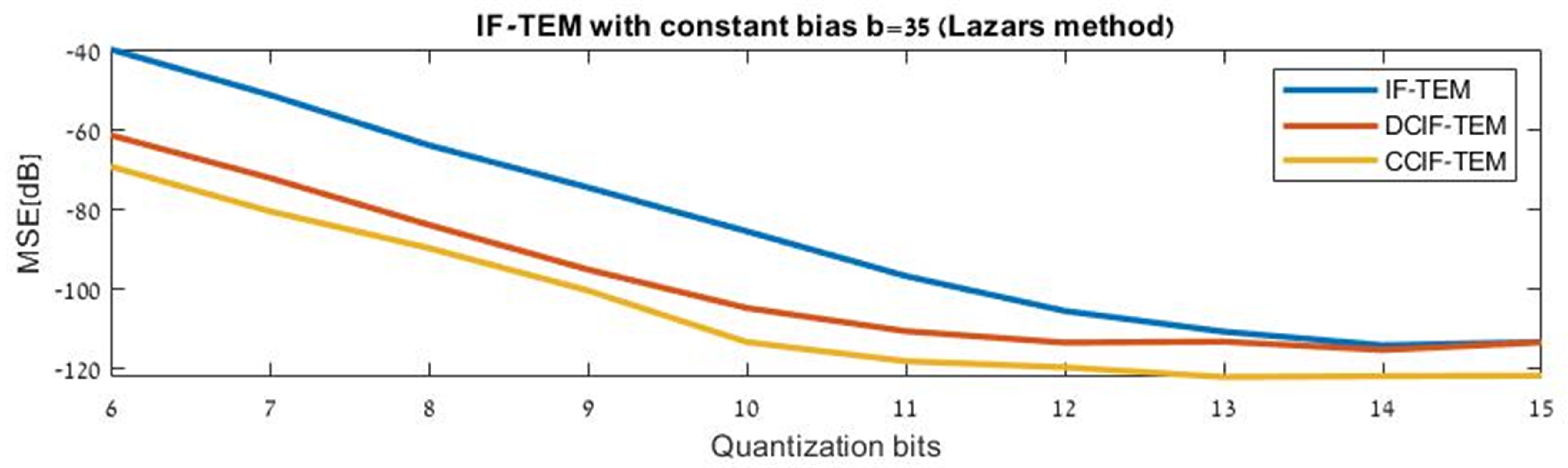}}
  \centerline{(b)}\medskip
\end{minipage}
\caption{Comparison in terms of MSE in dB between IF-TEM, CCIF-TEM, and DCIF-TEM: (a) Comparison between IF-TEM systems where $b=\alpha c$ and $\alpha=6$, (b) Comparison between IF-TEM systems where bias is constant, $b=35$.}
\label{fig:res1}
\end{figure}

In Fig. \ref{fig:res2} and Table 1, we show the overhead bits which are required for encoding the window numbers in CIF-TEM.
On average, only 5$\%$ additional bits are required to encode window number values. In addition, it can be observed that, on average, DCIF-TEM results in a lower overhead.
Let $p$ be defined as the overhead in the percentage of bits used in the CIF-TEM sampler, for a given $K$-levels quantizer. Since IF-TEM systems require two additional bits to achieve the same MSE as CIF-TEM systems, we can compute the compression obtained by using CIF-TEM as follows:
\[
\frac{\log_2(K-2)}{\log_2K}\times 100p.
\]
Table 1 presents the compression percentage in bits in using the CIF-TEM system while preserving the same MSE as in IF-TEM, i.e., Table 1 shows the compression in bits [$\%$] of encoding the entire signal using CIF-TEM systems with $K-2$ levels uniform quantizers compared to encoding the entire signal using IF-TEM sampler with $K$ level uniform quantizer. As seen in Table 1, for the same MSE in both types of systems, the compression in bits [$\%$] between the samplers ranges between 10$\%$ and 20$\%$. Consequently, the CIF-TEM sampler necessitates the use of extra bits to encode the window numbers of the samples, however, it reduces the overall number of bits that encode the signal.  Note that, increasing the number of quantization levels $K$, the compression we gain reduces. This happens because as we increase $K$, the MSE decreases and approaches its saturation value.

\begin{figure}[t]
\begin{minipage}[b]{1.0\linewidth}
  \centering
  \centerline{\includegraphics[trim=0cm 0cm 0cm 0cm, clip, width = 1.05\columnwidth]{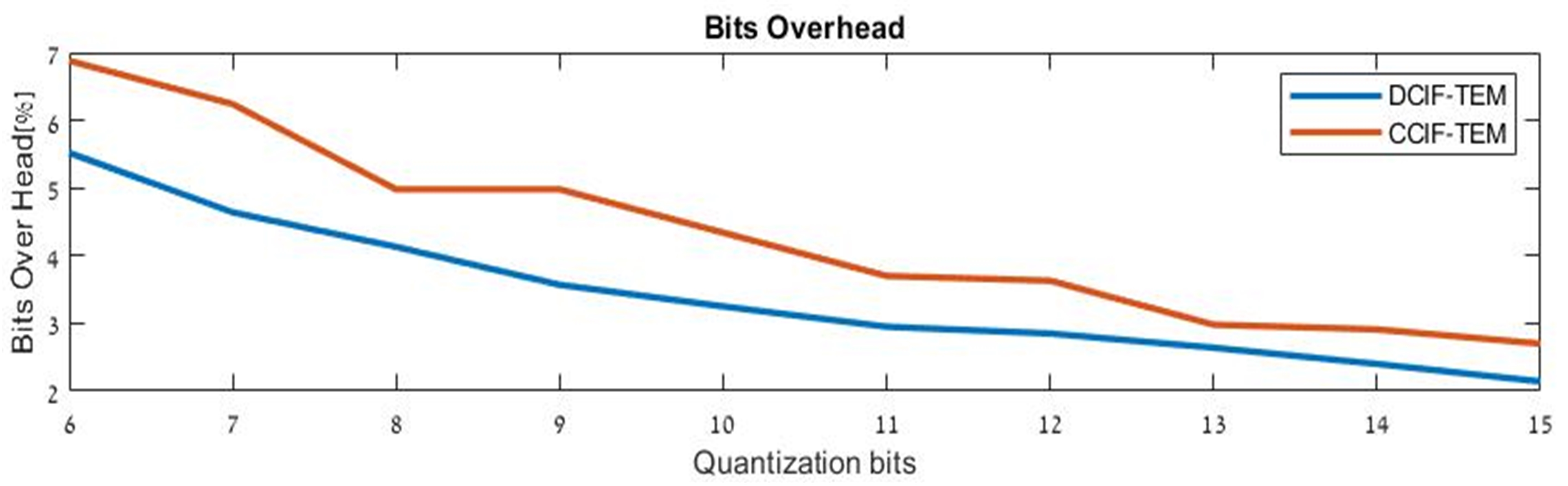}}
\end{minipage}
\caption{Average overhead (in bits) to encode the window number, $w_i$, in CIF-TEM systems. The overhead is presented in percentages.}
\label{fig:res2}
\end{figure}
\begin{table}[h!]
\begin{center}
\begin{tabular}{|c|c|c|c|}
\hline
\multicolumn{2}{|c|}{Quantizer bits} & \multicolumn{2}{|c|}{Overall system compression [$\%$]}\\ \hline
  CIF-TEM & IF-TEM & DCIF-TEM & CCIF-TEM\\
  \hline
  6 & 8 & 20.86\% & 19.84\%  \\
  \hline
  7 & 9 & 18.61\% & 17.37\% \\
  \hline
  8 & 10 & 16.7\% & 16.02\% \\
  \hline
  9 & 11 & 15.26\% & 14.11\% \\
  \hline
  10 & 12 & 13.96\% & 13.05\% \\
  \hline
  11 & 13 & 12.89\% & 12.25\% \\
  \hline
  12 & 14 & 11.84\% & 11.17\% \\
  \hline
  13 & 15 & 11.05\% & 10.75\% \\
  \hline
\end{tabular}
\caption{System compression (in bits $\%$). That is, the number of overall bits used to encode a signal using CIF- TEM scheme proposed, compared to the overall number of bits needed to encode the same signal using an IF-TEM scheme.}
\label{table:1}
\end{center}
\end{table}


\section{Conclusion}
\label{conclution}
Using the IF-TEM sampler for BL signals, we introduced methods for time encoding and decoding with analog compression.
Prior to the quantizer, analog compression is performed based on the stationarity of the encoded signal, a basic aspect of IF-TEM processing.
Low-bit-rate reconstruction is accomplished by subdividing the known IF-TEM dynamic range into tighter windows, which can be of fixed or variable size, and determining which window the sample resides within.
We empirically demonstrate that utilizing the same number of samples and up to 7$\%$ more bits than the traditional IF-TEM improves MSE by 5-20dB.
Using the compressed IF-TEM enables the usage of 1-2 less bits compared to the traditional IF-TEM, given the same reconstruction MSE target and amount of samples.
\bibliographystyle{IEEEtran}
\bibliography{refs}

\begin{thebibliography}{10}
\providecommand{\url}[1]{#1}
\csname url@samestyle\endcsname
\providecommand{\newblock}{\relax}
\providecommand{\bibinfo}[2]{#2}
\providecommand{\BIBentrySTDinterwordspacing}{\spaceskip=0pt\relax}
\providecommand{\BIBentryALTinterwordstretchfactor}{4}
\providecommand{\BIBentryALTinterwordspacing}{\spaceskip=\fontdimen2\font plus
\BIBentryALTinterwordstretchfactor\fontdimen3\font minus
  \fontdimen4\font\relax}
\providecommand{\BIBforeignlanguage}[2]{{%
\expandafter\ifx\csname l@#1\endcsname\relax
\typeout{** WARNING: IEEEtran.bst: No hyphenation pattern has been}%
\typeout{** loaded for the language `#1'. Using the pattern for}%
\typeout{** the default language instead.}%
\else
\language=\csname l@#1\endcsname
\fi
#2}}
\providecommand{\BIBdecl}{\relax}
\BIBdecl

\bibitem{miskowicz2006dynamic}
M.~Mi{\'s}kowicz and D.~Ko{\'s}cielnik, ``{The dynamic range of timing
  measurements of the asynchronous Sigma-Delta modulator},'' \emph{Proc. IFAC},
  vol.~39, no.~21, pp. 395--400, 2006.

\bibitem{koscielnik2007designing}
D.~Ko{\'s}cielnik and M.~Mi{\'s}kowicz, ``{Designing Time-to-Digital Converter
  for Asynchronous ADCs},'' \emph{Proc. IEEE Design Diag. Electron. Circuits
  Syst.}, pp. 1--6, 2007.

\bibitem{miskowicz2018event}
M.~Mi{\'s}kowicz, \emph{Event-based control and signal processing}.\hskip 1em
  plus 0.5em minus 0.4em\relax CRC press, 2018.

\bibitem{1199179}
E.~Allier, G.~Sicard, L.~Fesquet, and M.~Renaudin, ``{A new class of
  asynchronous A/D converters based on time quantization},'' in \emph{Ninth
  International Symposium on Asynchronous Circuits and Systems, 2003.
  Proceedings.}, 2003, pp. 196--205.

\bibitem{feichtinger2012approximate}
H.~G. Feichtinger, J.~C. Pr{\'\i}ncipe, J.~L. Romero, A.~S. Alvarado, and G.~A.
  Velasco, ``Approximate reconstruction of bandlimited functions for the
  integrate and fire sampler,'' \emph{Adv. Comput. Math.}, vol.~36, no.~1, pp.
  67--78, 2012.

\bibitem{adam2020sampling}
{K. Adam, A. Scholefield, and M. Vetterli}, ``Sampling and reconstruction of
  bandlimited signals with multi-channel time encoding,'' \emph{IEEE Tran.
  Signal Process.}, vol.~68, pp. 1105--1119, 2020.

\bibitem{lazar2004perfect}
A.~A. Lazar and L.~T. T{\'o}th, ``Perfect recovery and sensitivity analysis of
  time encoded bandlimited signals,'' \emph{IEEE Trans. Circuits Syst. I},
  vol.~51, no.~10, pp. 2060--2073, 2004.

\bibitem{lazar2004time}
A.~A. Lazar, ``Time encoding with an integrate-and-fire neuron with a
  refractory period,'' \emph{Neurocomputing}, vol.~58, pp. 53--58, 2004.

\bibitem{alexandru2019reconstructing}
R.~Alexandru and P.~L. Dragotti, ``Reconstructing classes of non-bandlimited
  signals from time encoded information,'' \emph{IEEE Trans. Signal Process.},
  vol.~68, pp. 747--763, 2019.

\bibitem{naaman2022fri}
H.~Naaman, S.~Mulleti, and Y.~C. Eldar, ``{FRI-TEM: Time encoding sampling of
  finite-rate-of-innovation signals},'' \emph{IEEE Transactions on Signal
  Processing}, 2022.

\bibitem{lazar2011video}
A.~A. Lazar and E.~A. Pnevmatikakis, ``Video time encoding machines,''
  \emph{IEEE Transactions on Neural Networks}, vol.~22, no.~3, pp. 461--473,
  2011.

\bibitem{thao2020time}
N.~T. Thao and D.~Rzepka, ``{Time encoding of bandlimited signals:
  Reconstruction by pseudo-inversion and time-varying multiplierless FIR
  filtering},'' \emph{IEEE Transactions on Signal Processing}, vol.~69, pp.
  341--356, 2020.

\bibitem{adam2021asynchrony}
K.~Adam, A.~Scholefield, and M.~Vetterli, ``Asynchrony increases efficiency:
  Time encoding of videos and low-rank signals,'' \emph{IEEE Transactions on
  Signal Processing}, vol.~70, pp. 105--116, 2021.

\bibitem{naaman2021time}
H.~Naaman, S.~Mulleti, Y.~C. Eldar, and A.~Cohen, ``{Time-Based Quantization
  for FRI and Bandlimited Signals},'' in \emph{2022 30th European Signal
  Processing Conference (EUSIPCO)}.\hskip 1em plus 0.5em minus 0.4em\relax
  IEEE, 2022, pp. 2241--2245.

\bibitem{naaman2022uniqueness}
H.~Naaman, S.~Mulleti, and Y.~C. Eldar, ``{Uniqueness and Robustness of
  Tem-Based FRI Sampling},'' in \emph{2022 IEEE International Symposium on
  Information Theory (ISIT)}.\hskip 1em plus 0.5em minus 0.4em\relax IEEE,
  2022, pp. 2631--2636.

\bibitem{rudresh2020time}
{S. Rudresh, A. J. Kamath, and C. S. Seelamantula}, ``{A Time-Based Sampling
  Framework for Finite-Rate-of-Innovation Signals},'' in \emph{Proc. IEEE Int.
  Conf. Acoust., Speech and Signal Process. (ICASSP)}, 2020, pp. 5585--5589.

\bibitem{gontier2014sampling}
D.~Gontier and M.~Vetterli, ``{Sampling based on timing: Time encoding machines
  on shift-invariant subspaces},'' \emph{Applied and Comput. Harmonic Anal.},
  vol.~36, no.~1, pp. 63--78, 2014.

\bibitem{papoulis1967limits}
A.~Papoulis, ``Limits on bandlimited signals,'' \emph{Proc. of the IEEE},
  vol.~55, no.~10, pp. 1677--1686, 1967.

\bibitem{popoviciu1935equations}
T.~Popoviciu, ``Sur les {\'e}quations alg{\'e}briques ayant toutes leurs
  racines r{\'e}elles,'' \emph{Mathematica}, vol.~9, no. 129-145, p.~20, 1935.

\end{thebibliography}
\end{document}